\documentclass[aps,preprint]{revtex4}%
\usepackage{amsfonts}
\usepackage{amsmath}
\usepackage{amssymb}
\usepackage{graphicx}%
\providecommand{\U}[1]{\protect\rule{.1in}{.1in}}
\newtheorem{theorem}{Theorem}

\newtheorem{proposition}[theorem]{Proposition}

\newenvironment{proof}[1][Proof]{\noindent\textbf{#1.} }{\ \rule{0.5em}{0.5em}}
\begin{document}
\title[Short title for running header]{Coherent states, vacuum structure and infinite component relativistic wave
equations }
\author{Diego Julio Cirilo-Lombardo}
\affiliation{National Institute of Plasma Physics (INFIP), Consejo Nacional de
Investigaciones Cientificas y Tecnicas (CONICET), Facultad de Ciencias Exactas
y Naturales, Universidad de Buenos Aires, Ciudad Universitaria, Buenos Aires
1428, Argentina}
\affiliation{Bogoliubov Laboratory of Theoretical Physics, 141980, Joint Institute for
Nuclear Research, Dubna (Moscow Region), Russian Federation}
\pacs{PACS number}

\begin{abstract}
It is commonly claimed in the recent literature that certain solutions to wave
equations of positive energy of Dirac-type with internal variables are
characterized by a non-thermal spectrum. As part of that statement, it was
said that the transformations and symmetries involved in equations of such
type correspond to a particular representation of the Lorentz group. In this
paper we give the general solution to this problem emphasizing the interplay
between the group structure, the corresponding algebra and the physical
spectrum. This analysis is completed with a strong discussion and proving
that: i) the physical states are represented by coherent states; ii) the
solutions in previous references [1] are not general, ii) the symmetries of
the considered physical system in [1] (equations and geometry) do not
correspond to the Lorentz group but to the fourth covering: the Metaplectic
group $Mp(n)$.

\end{abstract}
\volumeyear{year}
\volumenumber{number}
\issuenumber{number}
\eid{identifier}
\date[Date text]{date}
\received[Received text]{date}

\revised[Revised text]{date}

\accepted[Accepted text]{date}

\published[Published text]{date}

\startpage{101}
\endpage{102}
\maketitle
\tableofcontents

\section{Introduction and results}

For the last 50 years, there has been an increase in the interest with respect
to two fundamental points of theoretical physics: the new representations of
algebras with variables of the harmonic oscillator and the study of
relativistic wave equations. These two points were developed placing great
attention on the condition of positive energy and the role of the spin in
these representations. The main motives were the theoretical problems of
optics, the positive energy spectrum of physical states and the close relation
between the spin and the generalized statistics. Despite the recent interest
and the continuous efforts into the study on compact groups and their
relationship to physics, there was no major progress on the issue.

For example, in a recent reference it was claimed that certain solutions to
wave equations of positive energy of the Dirac type with internal variables
have as the main characteristic a non-thermal spectrum. As part of that
statement, it was said that the transformations and symmetries involved in a
such an equation with correspond to a particular representation of the Lorentz group.

In this work we will demonstrate that both claims and the same state solutions
in [1] are unfortunately not fully correct. Calculating the solutions of the
physical system of [1] we show that:

i) the solutions are coherent states, as described before (e.g.[2,5,6,7,8,9]);

ii) we show that the transformations and symmetries involved into the
relativistic wave equation of [1] do not belong to the group of Lorentz but to
the double cover of the groups$Sp(2)$ and $SU(1,1$): the Metaplectic group
$Mp(2)$ [2,6,7];

iii) and that these solutions have, in general, a thermal spectrum going under
certain conditions, to the non-classical behaviour (squeezed)[1,5].

Regarding the theoretical basis of the problem we start as follows:

Let such a spinor such that can be described schematically by the chain:
\begin{equation}
A_{\alpha}:\in Mp\left(  2\right)  \supset Sp\left(  2\mathbb{R}\right)  \sim
SU(1,1)\supset SO(1,2)\approx L\left(  3\right)  \tag{1}%
\end{equation}
(take not of the above structure that will be important into the analysis that
follows) that is defined as:%

\begin{equation}
A_{\alpha}=\left(
\begin{array}
[c]{c}%
a\\
a^{+}%
\end{array}
\right)  _{\beta}\Rightarrow\left[  A_{\alpha},A_{\beta}\right]
=\epsilon_{\alpha\beta} \tag{2}%
\end{equation}
where $a$ and $a^{+}$ are standard annihilation and creation operators,
respectively. As we will see soon, there exists a close relation with the
squeezed vacuum structure. The equation to solve has the typical structure of
the positive energy equation with internal variables, as proposed by
Majorana[4] and Dirac[5], and is explicitly written as
\begin{equation}
\left(  \sigma^{i}\partial_{i}-m\right)  _{\alpha}^{\beta}A_{\beta}\left\vert
\psi\right\rangle =0 \tag{3}%
\end{equation}
In reference [1], similarly to the case of the Dirac positive energy equation,
a wave solution was proposed as:%
\begin{equation}
\left\vert \psi\right\rangle =e^{ip\cdot x}\left\vert u\right\rangle \tag{4}%
\end{equation}
The first wrong fact in ref.[1] is to assume a priori that the momentum $p$
and $x$ in the exponent of the proposed wave equation (4) commute with the
annihilation and creation operators $a$ and $a^{+}$. Consequently, in our
analysis we will to consider the phase space coordinates $p$ and $x$ in the
exponent of the proposed wave equations as constants or as if the annihilation
and creation operators $a$ and $a^{+}$ act in an internal or auxiliary space.

Only under these conditions we can insert (4) in (3) obtaining:
\begin{align}
\left(  ip_{i}\sigma_{i}-m\right)  _{\alpha}^{\beta}\left(
\begin{array}
[c]{c}%
a\\
a^{+}%
\end{array}
\right)  _{\beta}\left\vert u\right\rangle  &  =0\tag{5}\\
\left(
\begin{array}
[c]{cc}%
ip_{3}-m & ip_{1}-p_{2}\\
ip_{1}+p_{2} & -ip_{3}-m
\end{array}
\right)  \left(
\begin{array}
[c]{c}%
a\\
a^{+}%
\end{array}
\right)  _{\beta}\left\vert u\right\rangle  &  =0\tag{6}%
\end{align}
At this point the second wrong fact in [1] is evident: the author remains with
only one component of the spinor solution. In fact if we impose the same
conditions as in [1] namely $p_{i}=\left(  0,p,i\varepsilon\right)  $ we have%
\begin{align}
\left(
\begin{array}
[c]{cc}%
\varepsilon+m & p\\
-p & m-\varepsilon
\end{array}
\right)  \left(
\begin{array}
[c]{c}%
a\\
a^{+}%
\end{array}
\right)  _{\beta}\left\vert u\right\rangle  &  =0\tag{7}\\
\left(
\begin{array}
[c]{c}%
\left(  \varepsilon+m\right)  a+pa^{+}\\
-pa+\left(  m-\varepsilon\right)  a^{+}%
\end{array}
\right)  _{\beta}\left\vert u\right\rangle  &  =0\tag{8}%
\end{align}
Notice that there are two different and simultaneous conditions that
$\left\vert u\right\rangle $ must satisfy. If we put now $p=0$ \ as in the
ref.[1] then%
\begin{equation}
\left(
\begin{array}
[c]{c}%
\left(  \varepsilon+m\right)  a\\
\left(  m-\varepsilon\right)  a^{+}%
\end{array}
\right)  _{\beta}\left\vert u_{0}\right\rangle =0\tag{9}%
\end{equation}
Here we clearly see that $\left\vert u_{0}\right\rangle $ cannot be the Fock
vacuum $\left\vert 0\right\rangle $ as stated ref [1] (it can only be if
$m=\varepsilon).$Through the next sections we will find the true vacuum, the
spectrum and the solution of the problem.

\section{Relation with the squeezed vacuum}

Looking at expressions (7,8)it is not difficult to see that these can be
obtained from similar form as the squeezed vacuum. The squeezed vacuum is
generated by the $Mp(2)$ transformation $U=S\left(  \xi\right)  $
\begin{equation}
A_{\alpha}\rightarrow S\left(  \xi\right)  \left(
\begin{array}
[c]{c}%
a\\
a^{+}%
\end{array}
\right)  _{\alpha}S^{\dagger}\left(  \xi\right)  =\left(
\begin{array}
[c]{c}%
\lambda a+\mu a^{+}\\
\lambda^{\ast}a^{+}+\mu^{\ast}a
\end{array}
\right)  _{\alpha}\tag{10}%
\end{equation}
where $\lambda\left(  \xi\right)  $ and $\mu\left(  \xi\right)  $ satisfy
$\left\vert \lambda\right\vert ^{2}+\left\vert \mu\right\vert ^{2}=1,$ e.g.
SU(1,1) elements.

\bigskip We must note that the RHS\ of eq. (10) is governed by the operators
$S\left(  \xi\right)  \in Mp\left(  2\right)  $ being the right side affected
by a matrix representation of $SU\left(  1,1\right)  $ as follows%
\begin{equation}
S\left(  \xi\right)  \left(
\begin{array}
[c]{c}%
a\\
a^{+}%
\end{array}
\right)  _{\alpha}S^{\dagger}\left(  \xi\right)  =\left(
\begin{array}
[c]{cc}%
\lambda & \mu\\
\mu^{\ast} & \lambda^{\ast}%
\end{array}
\right)  \left(
\begin{array}
[c]{c}%
a\\
a^{+}%
\end{array}
\right)  _{\alpha} \tag{11}%
\end{equation}

\bigskip

\bigskip

Clearly the above equivalence is only local (infinitesimal) since at the level
of the group structure (see the chain (1)) there is an homomorphism
relationship. The homomorphisms between $Mp\left(  2\right)  $ and $SU\left(
1,1\right)  $ (or $Sp\left(  2\mathbb{R}\right)  ),$ which are two to one and
four to one in the case of $SO\left(  1,2\right)  $, can be expressed in
$\alpha$ (polar)-parameterization [6,7] in the usual way:%

\begin{align}
S\left(  \alpha_{\perp},\alpha_{3}\right)   &  \in Mp\left(  2\right)
\rightarrow s\left(  \alpha_{\perp},\left[  \alpha_{3}\right]  _{4\pi}\right)
\in Sp\left(  2\mathbb{R}\right) \tag{12}\\
\text{ \ \ \ \ \ \ \ \ \ }  &  \rightarrow s\left(  \alpha_{\perp},\left[
\alpha_{3}\right]  _{4\pi}\right)  \in SU\left(  1,1\right)  \tag{13}%
\end{align}%
\begin{gather}
\left[  \alpha_{3}\right]  _{4\pi}\in(-2\pi,2\pi]\rightarrow\left[  \alpha
_{3}\right]  _{4\pi}\operatorname{mod}4\pi,\alpha_{\perp}\in\mathbb{R}%
^{2},\alpha_{3}\in(-4\pi,4\pi]\text{ for the }SU\left(  1,1\right)  (\text{or
}Sp\left(  2\mathbb{R}\right)  )\text{ case}\tag{14}\\
\rightarrow\left[  \alpha_{3}\right]  _{8\pi}\operatorname{mod}8\pi
,\alpha_{\perp}\in\mathbb{R}^{2},\alpha_{3}\in(-8\pi,8\pi]\text{ }\tag{15}\\
\text{for the }SO\left(  1,2\right)  \text{ case}\nonumber
\end{gather}

Consequently, it is clear that the "two to one" and the "two to four" nature
are involved in the reduction of the range of the parameter $\alpha_{3}.$ This
is the main reason why, for the physical scenarios of current interest, the
above parameterization is better than the Iwasawa (KAN)\ one.

The most general expression for an element of the Metaplectic group can be
computed, with the following result:
\[
e^{A\left(  aa^{+}+a^{+}a\right)  +Ba^{+2}+Ca^{2}}=e^{-A/2}\exp\left(
\frac{Ba^{+2}}{\Delta\coth\Delta-A}\right)  \exp\left[  H\ln\left(
\frac{\Delta\sec h\Delta}{\Delta-A\tanh\Delta}\right)  \right]  \exp\left(
\frac{Ca^{2}}{\Delta\coth\Delta-A}\right)  ,
\]%
\begin{equation}
\Delta\equiv\sqrt{A^{2}-4BC},H\equiv\left(  \frac{aa^{+}+a^{+}a}{2}\right)
,\widehat{N}\equiv a^{+}a \tag{16}%
\end{equation}
where the Baker-Haussdorf-Campbell formula was used. $A,B,C$ are arbitrary in
principle only linked by expression (16) (all group theoretical propierties of
\ the noncompact groups involved, were assumed there). Therefore, with the
parameters as given by expressions (7) to (9) $S\left(  \xi\right)  $ takes a
concrete form as follows
\[
S\left(  \xi\right)  =\exp\left(  \frac{p}{m+\epsilon}a^{+2}\right)  \left(
\frac{1}{\sqrt{m^{2}-\epsilon^{2}}}\right)  ^{1/2}\left\{  \underset
{n=0}{\overset{\infty}{\sum}}\frac{1}{n!}\left[  \ln\left(  \frac{1}%
{\sqrt{m^{2}-\epsilon^{2}}}\right)  \widehat{N}\right]  ^{n}\right\}
\exp\left(  -\frac{p}{m-\epsilon}a^{2}\right)  ,
\]
thus, the unitary (squeezed) operator acting on the true vacuum (fiducial
vector) defines the following general state%
\begin{equation}
\left\vert \xi\right\rangle \equiv S\left(  \xi\right)  \left\vert
z_{0}\right\rangle \tag{17}%
\end{equation}

\section{The solution}

We arrive to the construction of coherent states on a general vacuum:
$A\left\vert 0\right\rangle +B\left\vert 1\right\rangle $ with $A$\ and $B$
depending on initial and boundary conditions. If $\left\vert z_{0}%
\right\rangle \equiv A\left\vert 0\right\rangle +B\left\vert 1\right\rangle $
then
\begin{gather}
\left\vert \xi\right\rangle \equiv S\left(  \xi\right)  \left\vert
z_{0}\right\rangle =\frac{\exp\left(  \alpha a^{+2}\right)  }{(m^{2}%
-\epsilon^{2})^{1/4}}\left[  A\left\vert 0\right\rangle +\frac{B}%
{(m^{2}-\epsilon^{2})^{1/2}}\left\vert 1\right\rangle \right]  \tag{18}\\
=(m^{2}-\epsilon^{2})^{-1/4}\overset{\infty}{\underset{n=0}{\sum}}\frac
{\alpha^{n}}{n!}\left(  a^{+2}\right)  ^{n}\left[  A+\frac{B}{(m^{2}%
-\epsilon^{2})^{1/2}}a\right]  \left\vert 0\right\rangle ,\tag{19}\\
with\text{ \ \ \ \ }\alpha\equiv\frac{p/2}{m+\epsilon}\tag{20}%
\end{gather}
Notice that $a^{2}$ annihilates $\left\vert z_{0}\right\rangle $ but $a$ does
not. After standard normalization, the constants in the "thermal" (photon)
case reach the critical point that is when the quantum state solution is
simultaneously eigenstate of $a$ and $a^{2}$, they take the particular
fashion
\begin{align}
A &  =\left(  \left\vert m^{2}-\epsilon^{2}\right\vert +p^{2}sign\left(
\epsilon^{2}-m^{2}\right)  \right)  ^{1/4},\tag{21}\\
B &  =\left(  \left\vert m^{2}-\epsilon^{2}\right\vert +p^{2}sign\left(
\epsilon^{2}-m^{2}\right)  \right)  ^{3/4}=A^{3}\tag{22}%
\end{align}
We have a standard coherent state (eigenstate of the operator $a$ ) as a
linear combination of two states belonging to $\mathcal{H}_{1/4}$ and
$\mathcal{H}_{3/4}$ respectively (that are independent coherent states as
eigenvalues of $a^{2}$). In this particular case we have%
\begin{equation}
\left\vert z_{0}\right\rangle _{th}=A(1+A^{2}a^{+})\left\vert 0\right\rangle
\tag{23}%
\end{equation}
notice that this vacuum is not singular at $m\rightarrow\epsilon$ but is
analytically continued into the complex plane where it is defined:%
\begin{equation}
\left\vert \xi\right\rangle _{th}\equiv S\left(  \xi\right)  \left\vert
z_{0}\right\rangle _{th}=\left(  1+\frac{p^{2}sign\left(  \epsilon^{2}%
-m^{2}\right)  }{\left\vert m^{2}-\epsilon^{2}\right\vert }\right)
^{1/4}e^{\frac{p/2}{m+\epsilon}a^{+2}}\left[  1+\left(  1+\frac{p^{2}%
sign\left(  \epsilon^{2}-m^{2}\right)  }{\left\vert m^{2}-\epsilon
^{2}\right\vert }\right)  ^{1/2}a^{+}\right]  \left\vert 0\right\rangle
\tag{24}%
\end{equation}

\section{Bargmann representation: analytical vs. geometrical viewpoint}

\subsection{The Bargmann representation}

We have so far worked mainly with the photon-number description of the Hilbert
space $\mathcal{H}$ and the operators $a,$ $a^{+}.$ In this section we analyze
the misunderstanding pointed out previously introducing the Bargmann representation.

The Bargmann representation of $\mathcal{H}$ associates an entire analytic
function$f\left(  z\right)  $ of a complex variable $z$, with each
vector$\left\vert \varphi\right\rangle \in\mathcal{H}$ in the following
manner:%
\begin{align}
\left\vert \varphi\right\rangle  &  \in\mathcal{H\rightarrow}f\left(
z\right)  =\underset{n=0}{\overset{\infty}{\sum}}\left\langle n\right.
\left\vert \varphi\right\rangle \frac{z^{n}}{\sqrt{n!}},\tag{25}\\
\left\langle \varphi\right.  \left\vert \varphi\right\rangle  &
\equiv\left\vert \left\vert \varphi\right\vert \right\vert ^{2}=\underset
{n=0}{\overset{\infty}{\sum}}\left\vert \left\langle n\right.  \left\vert
\varphi\right\rangle \right\vert ^{2}\tag{26}\\
&  =\int\frac{d^{2}z}{\pi}e^{-\left\vert z\right\vert ^{2}}\left\vert f\left(
z\right)  \right\vert ^{2},\tag{27}%
\end{align}
where the integration is over the entire complex plane. The above association
can be compactly written in terms of the normalized coherent states of the
Barut-Girardello type, namely, (right) eigenstates of the annihilation
operator $a:$%
\begin{align}
a\left\vert z\right\rangle  &  =z\left\vert z\right\rangle ,\tag{28}\\
\left\vert z\right\rangle  &  =e^{-\left\vert z\right\vert ^{2}/2}%
\underset{n=0}{\overset{\infty}{\sum}}\frac{z^{n}}{\sqrt{n!}}\left\vert
n\right\rangle ,\tag{29}\\
\left\langle z^{\prime}\right.  \left\vert z\right\rangle  &  =e^{\left(
-\left\vert z\right\vert ^{2}/2-\left\vert z^{\prime}\right\vert
^{2}/2+z^{\prime\ast}z\right)  }\tag{30}%
\end{align}
then, we have%
\begin{equation}
f\left(  z\right)  =e^{-\left\vert z\right\vert ^{2}/2}\left\langle z^{\ast
}\right.  \left\vert \varphi\right\rangle \tag{31}%
\end{equation}
However, $f\left(  z\right)  $ must be as $\left\vert z\right\vert
\rightarrow\infty$ so that $\left\vert \left\vert \varphi\right\vert
\right\vert $ is finite. In this particular representation, the actions of $a$
and $a$ $^{+},$ and the functions representing $\left\vert n\right\rangle $
are as follows:%
\begin{align}
\left(  a^{+}f\right)  \left(  z\right)   &  =zf\left(  z\right)  ,\tag{32}\\
\left(  af\right)  \left(  z\right)   &  =\frac{df\left(  z\right)  }%
{dz},\tag{33}\\
\left\vert n\right\rangle  &  \rightarrow\frac{z^{n}}{\sqrt{n!}}\tag{34}%
\end{align}

\subsection{$Mp\left(  2\right)  $ generalized coherent states in the Bargmann
representation}

Having introduced the necessary ingredients, we can now describe the physical
states of the system under consideration.

i) The $\mathcal{H}_{1/4}$ states occupy the sector even of the full Hilbert
space $\mathcal{H}$ and we may describe them as follows
\begin{align}
f^{\left(  +\right)  }\left(  z,\omega\right)   &  =\left(  1-\left\vert
\omega\right\vert ^{2}\right)  ^{1/4}e^{\omega z^{2}/2}\tag{35}\\
&  =\left(  1-\left\vert \omega\right\vert ^{2}\right)  ^{1/4}\underset
{m=0,1,2,..}{\sum}\frac{\left(  \omega/2\right)  ^{m}}{m!}z^{2m} \tag{36}%
\end{align}
then, in the vector representation we have:%
\begin{equation}
\left\vert \Psi^{\left(  +\right)  }\left(  \omega\right)  \right\rangle
=\left(  1-\left\vert \omega\right\vert ^{2}\right)  ^{1/4}\underset
{m=0,1,2,..}{\sum}\frac{\left(  \omega/2\right)  ^{m}}{m!}\sqrt{2m!}\left\vert
2m\right\rangle \tag{37}%
\end{equation}
consequently, the number representation is obtained as:%
\begin{align}
\left\langle 2m\right.  \left\vert \Psi^{\left(  +\right)  }\left(
\omega\right)  \right\rangle  &  =\left(  1-\left\vert \omega\right\vert
^{2}\right)  ^{1/4}\frac{\left(  \omega/2\right)  ^{m}}{m!}\sqrt{2m!}%
\tag{38}\\
\left\langle 2m+1\right.  \left\vert \Psi^{\left(  +\right)  }\left(
\omega\right)  \right\rangle  &  \equiv0 \tag{39}%
\end{align}

\bigskip

\bigskip

\bigskip

\bigskip ii) The $\mathcal{H}_{3/4}$ states occupy the odd sector of the full
Hilbert space $\mathcal{H}$ and we may describe them as before:%
\begin{align}
f^{\left(  -\right)  }\left(  z,\omega\right)   &  =\left(  1-\left\vert
\omega\right\vert ^{2}\right)  ^{3/4}ze^{\omega z^{2}/2}\tag{39}\\
&  =\left(  1-\left\vert \omega\right\vert ^{2}\right)  ^{3/4}\underset
{m=0,1,2,..}{\sum}\frac{\left(  \omega/2\right)  ^{m}}{m!}z^{2m+1} \tag{40}%
\end{align}
then, in vector representation we have:%
\begin{equation}
\left\vert \Psi^{\left(  -\right)  }\left(  \omega\right)  \right\rangle
=\left(  1-\left\vert \omega\right\vert ^{2}\right)  ^{3/4}\underset
{m=0,1,2,..}{\sum}\frac{\left(  \omega/2\right)  ^{m}}{m!}\sqrt{\left(
2m+1\right)  !}\left\vert 2m+1\right\rangle \tag{41}%
\end{equation}
The number representation is consequently:%
\begin{align}
\left\langle 2m+1\right.  \left\vert \Psi^{\left(  -\right)  }\left(
\omega\right)  \right\rangle  &  =\left(  1-\left\vert \omega\right\vert
^{2}\right)  ^{3/4}\frac{\left(  \omega/2\right)  ^{m}}{m!}\sqrt{\left(
2m+1\right)  !}\tag{42}\\
\left\langle 2m\right.  \left\vert \Psi^{\left(  -\right)  }\left(
\omega\right)  \right\rangle  &  \equiv0 \tag{43}%
\end{align}

iii) The full Hilbert space, defined by the direct sum $\mathcal{H}%
=\mathcal{H}_{1/4}\oplus\mathcal{H}_{3/4},$ is trivially described as follows%
\begin{align}
f\left(  z,\omega\right)   &  =f^{\left(  +\right)  }\left(  z,\omega\right)
+f^{\left(  -\right)  }\left(  z,\omega\right) \tag{44}\\
&  =\left(  1-\left\vert \omega\right\vert ^{2}\right)  ^{1/4}\underset
{m=0,1,2,..}{\sum}\frac{\left(  \omega/2\right)  ^{m}}{m!}z^{2m}\left[
1+\left(  1-\left\vert \omega\right\vert ^{2}\right)  ^{1/2}z\right]  \tag{45}%
\end{align}
Then, in complete analogy as their even and odd subspaces, the corresponding
states are described by:%
\begin{align}
\Psi\left(  \omega\right)   &  =\Psi^{\left(  +\right)  }\left(
\omega\right)  +\Psi^{\left(  -\right)  }\left(  \omega\right) \tag{46}\\
&  =\left(  1-\left\vert \omega\right\vert ^{2}\right)  ^{1/4}\underset
{m=0,1,2,..}{\sum}\frac{\left(  \omega/2\right)  ^{m}}{m!}\sqrt{2m!}\left[
1+\left(  1-\left\vert \omega\right\vert ^{2}\right)  ^{1/2}a^{+}\right]
\left\vert 2m\right\rangle \tag{47}%
\end{align}%
\begin{equation}
\left\langle m\right.  \left\vert \Psi\left(  \omega\right)  \right\rangle
\left\{
\begin{array}
[c]{c}%
\left(  1-\left\vert \omega\right\vert ^{2}\right)  ^{1/4}\frac{\left(
\omega/2\right)  ^{m}}{m!}\sqrt{2m!},\left(  m\text{ even}\right) \\
\\
\left(  1-\left\vert \omega\right\vert ^{2}\right)  ^{3/4}\frac{\left(
\omega/2\right)  ^{m}}{m!}\sqrt{\left(  2m+1\right)  !},\left(  m\text{
odd}\right)
\end{array}
\right.  \tag{48}%
\end{equation}
where the link between the physical observables and the group parameters is
given by the following expression (measure):
\begin{equation}
\left(  1+\frac{p^{2}sign\left(  \epsilon^{2}-m^{2}\right)  }{\left\vert
m^{2}-\epsilon^{2}\right\vert }\right)  ^{1/4}\rightarrow\left(  1-\left\vert
\omega\right\vert ^{2}\right)  ^{1/4} \tag{49}%
\end{equation}

\section{The Limit $\epsilon\rightarrow m$}

This is precisely the limit $\left\vert \omega\right\vert ^{2}\rightarrow1$
from the point of view of the Metaplectic analysis that corresponds to the
edge of the complex disc. As we could see easily, the states solutions are
generally thermalized (full spectrum corresponding to $\mathcal{H}$). What
happens is that in the limit $\epsilon\rightarrow m$ the density of states
corresponding to $\mathcal{H}_{1/4}$ is greater than that of the odd states
belonging to $\mathcal{H}_{3/4}.$ It is for this reason that states belonging
to $\mathcal{H}_{1/4}$, will survive in this limit. As we will see in a
separate publication, that there is a particular case of the two-dimensional
electron transport with a magnetic field in the plane whose states belong to
metaplectic group.

\section{Complete equivalence between Sannikov's representation and
Metaplectic one}

The main characteristics of the particular representation introduced in [2] is
the following commutation relation that defines the generators $L_{i}:$%

\begin{equation}
\left[  L_{i},a^{\alpha}\right]  =\frac{1}{2}a^{\beta}\left(  \sigma
_{i}\right)  _{\beta}^{\text{ \ }\alpha} \tag{50}%
\end{equation}
The above representation which corresponds to a non-compact Lie algebra with
the following matrix form [5,7]is:%

\begin{equation}
\sigma_{i}=i\left(
\begin{array}
[c]{cc}%
0 & 1\\
1 & 0
\end{array}
\right)  ,\quad\sigma_{j}=\left(
\begin{array}
[c]{cc}%
0 & 1\\
-1 & 0
\end{array}
\right)  ,\quad\sigma_{k}=\left(
\begin{array}
[c]{cc}%
1 & 0\\
0 & -1
\end{array}
\right)  , \tag{51,52,53}%
\end{equation}
that fulfils evidently:
\begin{align}
\sigma_{i}\wedge\sigma_{j}  &  =-i\sigma_{k}\tag{54}\\
\sigma_{k}\wedge\sigma_{i}  &  =i\sigma_{j}\tag{55}\\
\sigma_{j}\wedge\sigma_{k}  &  =i\sigma_{i} \tag{56}%
\end{align}
The equivalence that we want to remark is manifested by the following:

\begin{proposition}
the generators in the representation of [2] fulfil:%
\begin{equation}
L_{i}=\frac{1}{2}a^{\beta}\left(  \sigma_{i}\right)  _{\beta}^{\text{
\ }\alpha}a_{\alpha}=T_{i} \tag{57}%
\end{equation}
where T$_{i}$ are the Metaplectic generators namely[6,7]:%
\begin{align}
T_{1}  &  =\frac{i}{4}\left(  a^{+2}-a^{2}\right) \tag{58}\\
T_{2}  &  =\frac{-1}{4}\left(  a^{+2}+a^{2}\right) \tag{59}\\
T_{3}  &  =\frac{-1}{4}\left(  aa^{+}+a^{+}a\right)  \tag{60}%
\end{align}

\end{proposition}

\begin{proof}
Explicitly in matrix form we can write the generators proposed in the paper
[2] (and for instance in[1]) as%
\begin{equation}
L_{i}=\overline{u}\mathbb{M}_{i}v \tag{61}%
\end{equation}%
\begin{align}
\overline{u}  &  \equiv\left(
\begin{array}
[c]{cc}%
a^{+} & a
\end{array}
\right) \tag{62}\\
v  &  \equiv\left(
\begin{array}
[c]{c}%
a\\
a^{+}%
\end{array}
\right)  \tag{63}%
\end{align}
In the representation (50), that is faithful, and taking into account that
$\sigma_{k}$ enter as "metric"in the sense given by Sannikov[2] we have%
\begin{align}
M_{1}  &  =\frac{i}{4}\left(
\begin{array}
[c]{cc}%
0 & 1\\
-1 & 0
\end{array}
\right)  =\frac{1}{4}\sigma_{k}\sigma_{i}\quad\tag{64}\\
M_{2}  &  =-\frac{1}{4}\left(
\begin{array}
[c]{cc}%
0 & 1\\
1 & 0
\end{array}
\right)  =-\frac{1}{4}\sigma_{k}\sigma_{j}\quad\tag{65}\\
M_{3}  &  =-\frac{1}{4}\left(
\begin{array}
[c]{cc}%
1 & 0\\
0 & 1
\end{array}
\right)  =-\frac{1}{4}\sigma_{k}^{2} \tag{66}%
\end{align}
consequently and by inspection (50)\ coincides with (61): thus, the
equivalence (57)\ is proved.
\end{proof}

\section{Concluding remarks}

In this paper, we have studied from the physical and group-theoretical point
of view, the close relation between the Metaplectic group, the Lorentz group
and its covering the $SL(2,C)$ ones. The main emphasis was put to clarify the
existent confusion between the representations of the considered non-compact
groups. To this end, using a typical example, a recently posed problem in [1],
we solved exactly the corresponding equations to the physical scenario given
in [1] ,highlighting consequently the common errors and misunderstandings that
appear to confuse representations: namely, the Metaplectic one with the other
non-compact (Lorentz and Special Linear). The analysis was made easier using
the group generators written with the Harmonic oscillator variables, arriving
at the following conclusions and results:

i) the solutions are coherent states, coinciding with previous theoretical
descriptions (e.g.[7,8,9]);

ii) the transformations and symmetries involved in the equation of [1] do not
belong to the group of Lorentz but to the double cover of $Sp(2)$ and
$SU(1,1)$: the Metaplectic group Mp (2)[2,5,7];

iii) and that these solutions are generally thermal going under certain
conditions to the non-classical condition (squeezed), as was verified before [5].

\section{Acknowledgements}

I am very grateful to the CONICET-Argentina and also to the BLTP-JINR
Directorate for their hospitality and finnantial support for part of this work.

\section{References}

1 Stepanovsky Yu. P.: Nucl.Phys.B (Proc.Suppl.), 102\&103, (2001) 407-411.

2 Sannikov S.S.: JETP 49, (1965) 1913 (in Russian).

3 Majorana E.: Nuovo Cimento 9, (1932) 335.

4 Dirac P.A.M. :Proc. Roy. Soc. A322 (1971) 435.

5 Cirilo-Lombardo D. J.: Physics Letters B 661, 186-191 (2008); Foundations of
Physics 37: 919-950 (2007); Found Phys (2009) 39: 373--396; The European
Physical Journal C - Particles and Fields, 2012, Volume 72, Number 7, 2079 ;
Cirilo-Lombardo Diego Julio and Afonso V.I.: Phys.Lett. A376 (2012) 3599-3603
; Cirilo-Lombardo D. J. and Prudencio T.: Int.J.Geom.Meth.Mod.Phys. 11 (2014) 1450067.

6. Arvind, Biswadeb Dutta, Mehta C. L., and Mukunda N.: Phys. Rev. A 50, 39 (1994)

7. Jordan T. F., Mukunda N. and Pepper S. V.: J. Math. Phys. 4, 1089 (1963).

8. Perelomov A.: Generalized Coherent States and Their Applications, Texts and
Monographs in Physics, Springer Berlin Heidelberg, 1986

9. Klauder, J. R. and Skagerstam B.-S.: Coherent States -- Applications in
Physics and Mathematical Physics (World Scientific, Singapore, 1985);

10. Klauder, J. R. and Sudarshan, E. C. G.: Fundamentals of Quantum Optics
(Benjamin, NY, 1968).

\end{document}